\documentclass[11pt]{article}

\usepackage{amssymb,amsthm,amsmath}
\usepackage{graphicx}
\usepackage[small]{caption}
\usepackage{subcaption}
\usepackage{epsfig}
	
\newcommand{\later}[1]{}
\newcommand{\old}[1]{}

\old{
\newcommand*\patchAmsMathEnvironmentForLineno[1]{
  \expandafter\let\csname old#1\expandafter\endcsname\csname #1\endcsname
  \expandafter\let\csname oldend#1\expandafter\endcsname\csname end#1\endcsname
  \renewenvironment{#1}
  {\linenomath\csname old#1\endcsname}
  {\csname oldend#1\endcsname\endlinenomath}}
  \newcommand*\patchBothAmsMathEnvironmentsForLineno[1]{
  \patchAmsMathEnvironmentForLineno{#1}
  \patchAmsMathEnvironmentForLineno{#1*}}
  \AtBeginDocument{
  \patchBothAmsMathEnvironmentsForLineno{equation}
  \patchBothAmsMathEnvironmentsForLineno{align}
  \patchBothAmsMathEnvironmentsForLineno{flalign}
  \patchBothAmsMathEnvironmentsForLineno{alignat}
  \patchBothAmsMathEnvironmentsForLineno{gather}
  \patchBothAmsMathEnvironmentsForLineno{multline}
}
} 
\usepackage{amsfonts}
\usepackage{booktabs}

\usepackage[utf8]{inputenc}
\usepackage{fullpage}
\usepackage{framed}

\usepackage{enumerate}
\usepackage{url}
\usepackage{hyperref}

\newtheorem{theorem}{Theorem}
\newtheorem{lemma}{Lemma}

\newcommand{\etal}{{et~al.}}
\newcommand{\ie}{{i.e.}}
\newcommand{\eg}{{e.g.}}

\newcommand{\NN}{\mathbb{N}} 
\newcommand{\eps}{\varepsilon}

\def\Prob{{\rm Prob}}
\def\E{{\rm E}}
\def\Var{{\rm Var}}

\title{\textsc{Finding a Mediocre Player}\footnote{
    A preliminary version of this paper appeared in the
\emph{Proceedings of the 11th International Conference on Algorithms and Complexity} 
(CIAC 2019), LNCS~11485, Springer, 2019, pp.~212--223.}}

\author{
Adrian Dumitrescu\thanks{%
Department of Computer Science,
University of Wisconsin--Milwaukee, USA\@.
Email~\texttt{dumitres@uwm.edu}.}}

\begin{document}

\maketitle

\begin{abstract}
  Consider a totally ordered set $S$ of $n$ elements; as an example, a set of tennis players
  and their rankings. Further assume that their ranking is a total order and thus satisfies
  transitivity and anti-symmetry.
  Following Frances Yao (1974), an element (player) is said to be $(i,j)$-\emph{mediocre}
  if it is neither among the top $i$ nor among the bottom $j$ elements of $S$.
  Finding a mediocre element is closely related to finding the median element.
  More than $40$ years ago, Yao suggested a very simple and elegant algorithm
  for finding an $(i,j)$-mediocre element: Pick $i+j+1$ elements arbitrarily 
  and select the $(i+1)$-th largest among them. She also asked:
  ``Is this the best algorithm?'' No one seems to have found a better algorithm ever since.

We first provide a deterministic algorithm that beats the worst-case comparison bound in
Yao's algorithm for a large range of values of $i$ (and corresponding suitable $j=j(i)$)
even if the current best selection algorithm is used. 
We then repeat the exercise for randomized algorithms; the average number of comparisons
of our algorithm beats the average comparison bound in Yao's algorithm for another large
range of values of $i$ (and corresponding suitable $j=j(i)$)
even if the best selection algorithm is used; the improvement is most notable
in the symmetric case $i=j$.
Moreover, the tight bound obtained in the analysis of Yao's algorithm allows
us to give a definite answer for this class of algorithms.
In summary, we answer Yao's question as follows:
(i)~``Presently not'' for deterministic algorithms and
(ii)~``Definitely not'' for randomized algorithms.
(In fairness, it should be said however that Yao posed the question in the context
of deterministic algorithms.)


\end{abstract}

\section{Introduction} \label{sec:intro}

Consider a totally ordered set $S$ of $n$ elements.    
Following Yao\footnote{Throughout this paper we frequently use the name Yao to
refer to Foong Frances Yao; we also use the name A.~C.-C.~Yao to refer to Andrew Chi-Chih Yao.},
an element is said to be $(i,j)$-\emph{mediocre} if it is neither among the top
(\ie, largest) $i$ nor among the bottom (\ie, smallest) $j$ elements of $S$. 
The notion of mediocre element introduced by Yao is essentially an approximate solution
to the selection problem, which can be stated as follows: 
Given a sequence $A$ of $n$ elements and an integer (selection) parameter $1\leq i\leq n$, 
the \emph{selection} problem is that of finding the $i$-th smallest element in $A$.
If the $n$ elements are distinct, the $i$-th smallest is larger than
$i-1$ elements of $A$ and smaller than the other $n-i$ elements of $A$.
By symmetry, the problems of determining the $i$-th smallest and the $i$-th largest
are equivalent. 

Together with sorting, the selection problem is one of the most fundamental problems
in computer science. Sorting trivially solves the selection problem; however, 
a higher level of sophistication is required in order to obtain a deterministic
linear time algorithm.  This was accomplished in the early 1970s, when Blum~\etal~\cite{BFP+73}
gave a $O(n)$-time algorithm for the problem. Their algorithm performs at most $5.43 n$
comparisons and its running time is linear irrespective of the selection parameter $i$.
Their approach was to use an element in $A$ as a pivot to partition $A$
into two smaller subsequences and recurse on one of them with a (possibly different)
selection parameter~$i$.  The pivot was set as the (recursively computed) median of
medians of small disjoint groups of the input array (of constant size at least $5$).
More recently, several variants of {\sc Select} with groups of $3$ and $4$,
also running in $O(n)$ time, have been obtained by Chen and Dumitrescu
and independently by Zwick~\cite{CD20}. On the other hand, a randomized linear time
algorithm for selection is relatively easy to obtain, simply by using a random pivot 
at each partitioning stage.

The selection problem and computing the median in particular are in close relation
with the problem of finding the quantiles of a set. 
The $k$-th \emph{quantiles} of an $n$-element set are the $k-1$ order statistics that divide
the sorted set in $k$ equal-sized groups (to within $1$); see, \eg, \cite[p.~223]{CLRS09}.
The $k$-th quantiles of a set can be computed by a recursive algorithm running
in $O(n \log{k})$ time. For $k=100$ the quantiles are called \emph{percentiles}. 

In an attempt to drastically reduce the number of comparisons done for selection
(down from $5.43 n$), Sch\"{o}nhage~\etal~\cite{SPP76} designed a non-recursive
algorithm based on different principles, most notably the technique of mass production.
Their algorithm finds the median (the $\lceil n/2 \rceil$-th largest element) using
at most $3n +o(n)$ comparisons; as noted by Dor and Zwick~\cite{DZ96},
it can be adjusted to find the $i$-th largest, for any $i$,
within the same comparison count. In a subsequent work,
Dor and Zwick~\cite{DZ99} managed to reduce the $3n +o(n)$ comparison bound to
about $2.95 n$; this however required new ideas and took a great deal of effort. 

\paragraph{Mediocre elements.}
Following Yao, an element is said to be $(i,j)$-\emph{mediocre} if it
is neither among the top (\ie, largest) $i$ nor among the bottom (\ie, smallest) $j$
of a totally ordered set $S$ of $n$ elements.
Yao remarked that historically finding a mediocre element
is closely related to finding the median, with a common motivation being selecting an
element that is not too close to either extreme. Observe also that $(i,j)$-mediocre elements
where $i=\lfloor \frac{n-1}{2} \rfloor$, $j= \lfloor \frac{n}{2} \rfloor$
(and symmetrically exchanged) are medians of~$S$. Intuitively, the larger the difference
$n-(i+j)$ is, the more candidates for an $(i,j)$-\emph{mediocre} exist and so the easier
the task of finding one should become. 

  In her PhD thesis~\cite{Yao74},  Yao suggested a very simple scheme (method)
  for finding an $(i,j)$-mediocre element: Pick $i+j+1$ elements arbitrarily 
  and select the $(i+1)$-th largest among them. It is easy to check that this element satisfies
  the required condition. It should be noted that this scheme becomes an algorithm when the
  selection algorithm is specified. By slightly abusing notation, we call this scheme an
  algorithm throughout this paper. 

  Yao asked whether this algorithm is optimal. No improvements over this algorithm were known.
  An interesting feature of this algorithm is that its complexity does not depend on $n$
  (unless $i$ or $j$ do).  Yao proved her algorithm is optimal for $i=1$.
  For $i+j +1 \leq n$, let $S(i,j,n)$ denote the minimum number of comparisons
  needed in the worst case to find an $(i,j)$-mediocre element.
  Yao~\cite[Sec.~4.3]{Yao74} proved that $S(1,j,n) = V_2(j+2) = j + \lceil \log(j+2) \rceil$,
  and so $S(1,j,n)$ is independent of $n$. Here $V_2(j+2)$ denotes the minimum number
  of comparisons needed in the worst case to find the second largest out of $j+2$ elements.
  
  The question of whether this algorithm is optimal for all values of $i$ and $j$ has remained
  open ever since. Another question is whether $S(i,j,n)$ is independent of $n$ for 
  other values of $i$ and $j$ (as is the case for $i=1$). Here we provide two alternative
  schemes for finding a mediocre element, one deterministic and one randomized, and thereby propose
  alternatives that can be compared and contrasted with Yao's algorithm.
  It should be pointed out that Yao's general question leads to several more specific
  questions:
  
\begin{enumerate} \itemsep 1pt
\item Assuming that an optimal deterministic algorithm for exact selection is available, is
  Yao's algorithm that uses this subroutine optimal?
\item Assuming that the current best deterministic algorithm for exact selection is used, is
  Yao's algorithm that uses this subroutine optimal?
\item How do the answers change if randomized algorithms are considered?
\end{enumerate}

Since recent progress on deterministic exact selection has been lagging
and it is quite possible that an optimal deterministic algorithm may be out of reach,
here we concentrate on the last two questions. 
However, it is perfectly possible that the answer to the first question is within reach
even if an optimal deterministic algorithm for selection has not been identified.

\paragraph{Background and related problems.}
Determining the comparison complexity for computing various order statistics including
the median has lead to many exciting questions, some of which are still unanswered today.
In this respect, Yao's hypothesis on selection~\cite{SPP76}, \cite[Sec.~4]{Yao74}
has stimulated the development of such algorithms~\cite{DZ96,Pa96,SPP76}.
That includes the seminal algorithm of Sch\"{o}nhage~\etal~\cite{SPP76}, which
introduced principles of mass production for deriving an efficient comparison-based algorithm.

Due to its primary importance, the selection problem has been studied extensively; see for
instance~\cite{BJ85,CM89,DHU+01,DZ96,DZ99,FR75,FG79,HS69,Ho61,KKZZ18,Ki81,Pa96,Yap76}.
A comprehensive review of early developments in selection is provided by Knuth~\cite{Kn98}.
The reader is also referred to dedicated book chapters on selection
in~\cite{AHU83,Ba88,CLRS09,DPV08} and the more recent articles~\cite{CD20,Du16},
including experimental work~\cite{Al17}. 

In many applications (\eg, sorting), it is not important to find an exact median
or any other precise order statistic for that matter, and an approximate median
suffices~\cite{EW18}. For instance, quick-sort type algorithms aim at finding a balanced partition
without much effort; see \eg, \cite{Ho61}. Various strategies that attempt to optimize
the sample size towards this goal have been studied in~\cite{MR01,MT95}. 

\paragraph{Our results.} Our main results are summarized in the two theorems
stated below. The comparison count of our deterministic algorithm for approximate selection
can be as low as about $89\%$ of the corresponding count for Yao's algorithm that employs
the current best deterministic algorithm for exact selection (for certain $i,j$ combinations).
Similarly, the comparison count of our randomized algorithm for approximate selection
can be as low as about $66\%$ of the corresponding count for the current best 
randomized algorithm for exact selection (for certain $i,j$ combinations).

\begin{theorem} \label{thm:deterministic}
  Given $n$ elements, an $(i,j)$-mediocre element where
  $i=\alpha n$, $j=(1-2\alpha) n -1$, and $0<\alpha<1/3$, can be found by a deterministic
  algorithm~A1 using $c_{\rm A1} \cdot n + o(n)$ comparisons in the worst case.
  If the number of comparisons done by Yao's algorithm is $c_{\rm Yao} \cdot n + o(n)$,
  we have $c_{\rm A1}< c_{\rm Yao}$ for each of the percentiles $1$ through $33$
  (\ie, $\alpha_s =s/100$, $s=1,\ldots,33$). 
  \end{theorem}

The constants $c_{\rm A1}=c_{\rm A1}(\alpha)$ and $c_{\rm Yao}=c_{\rm Yao}(\alpha)$ for the percentiles
$1$ through $33$ are given in Fig.~\ref{fig:tables}. See also Fig~\ref{fig:f2}. 

\begin{theorem} \label{thm:randomized}
  Given $n$ elements and $i,j \geq 1$ where $i+j+2(i+j)^{3/4} \leq n$ and $i+j =\omega(1)$,
  an $(i,j)$-mediocre element can be found by a randomized algorithm using
  $(i+j) + O((i+j)^{3/4})$ comparisons on average. 
\end{theorem}

For example: (i)~an $(i,j)$-mediocre element, where $i=j=n/2 - n^{3/4}$,
can be found using $n + O(n^{3/4})$ comparisons on average;
(ii)~if $\alpha,\beta>0$ are fixed constants with $\alpha +\beta<1$, 
an $(\alpha n, \beta n)$-mediocre element can be found using $(\alpha +\beta) n + O(n^{3/4})$
comparisons on average.

Note that finding an element near the median requires about $3n/2$ comparisons
for any previous algorithm (including Yao's), and finding the precise median requires
$3n/2 + o(n)$ comparisons on average, while the main term in this expression cannot
be improved ~\cite{CM89}. In contrast, our randomized algorithm finds an element
near the median in about $n$ comparisons on average, thereby achieving a substantial
savings of about $n/2$ comparisons.

\paragraph{Remarks.} 
Whereas the deterministic algorithm in Theorem~\ref{thm:deterministic} calls an
algorithm for exact selection as a subroutine, the randomized algorithm in
Theorem~\ref{thm:randomized} uses random sampling but does not use 
exact selection as a subroutine. It is worth recalling that presently no optimal algorithm
for exact selection is known with respect to the number of comparisons; \eg, the best known
bounds for median finding are as follows: this task can be accomplished with
$(2.95 + o(1))n$ comparisons by the algorithm of Dor and Zwick~\cite{DZ99}
and requires at least $(2+2^{-80})n$ comparisons by the result of the same authors~\cite{DZ01}.
On the other hand an optimal randomized algorithm for exact selection is known:
the $k$-th largest element out of $n$ given can be found using at most $n + \min(k,n-k) + o(n)$
comparisons on average by the algorithm of Floyd and Rivest~\cite{FR75} and  requires
$n + \min(k,n-k) + o(n)$ comparisons on average by the result of Cunto and Munro~\cite{CM89}.

\paragraph{Preliminaries and notation.} 
Without affecting the results, the following two standard simplifying assumptions
are convenient:
(i)~the input $A$ contains $n$ distinct elements; and 
(ii)~the floor and ceiling functions are omitted in the descriptions of the algorithms
and their analyses. 
For example if $\alpha \in (0,1)$, for simplicity we write the $\alpha n$-th element instead of
the more precise $\lfloor \alpha n \rfloor$-th (or $\lceil \alpha n \rceil$-th) element. 
In the same spirit, for convenience we treat $n^{1/4}$, $j \, n^{-1/4}$ (for $j \in \NN$)
and other algebraic expressions that appear in the description of the algorithms as integers.
Unless specified otherwise, all logarithms are in base $2$. 

Let $\E[X]$ and $\Var[X]$ denote the expectation and respectively, the variance,
of a random variable~$X$. If $E$ is an event in a probability space, $\Prob(E)$ denotes its
probability. Chebyshev's inequality is the following: For any $a>0$,
\begin{equation} \label{eq:C}
 \Prob\left(|X -\E[X]| \geq a \right) \leq \frac{\Var[X]}{a^2}.
\end{equation}
See for instance~\cite[p.~49]{MU05}.

\section{Deterministic approximate selection} \label{sec:deterministic}

Consider the problem of finding an $(i,j)$-mediocre element.
Without loss of generality (by considering the complementary order),
it can be assumed that $i \leq j$; and consequently $i<n/2$.
In addition, our algorithm is designed to work for a specific range $i \leq j \leq n-2i-1$
(hence $i<n/3$); outside this range our algorithm simply proceeds as in Yao's algorithm.
With anticipation, we note that our test values for purpose of comparison are contained
in the specified range.

Yao's algorithm is very simple:

\medskip
\noindent{\bf Algorithm~Yao.}
\begin{itemize} \itemsep 1pt
\item[] {\sc Step 1:} Choose an arbitrary subset of $i+j+1$ elements from the given $n$.

\item[] {\sc Step 2:} Select and return the $(i+1)$-th largest element from the chosen subset.
\end{itemize}

As mentioned earlier, it is easy to check that the element output by Yao's algorithm is
$(i,j)$-mediocre. Our algorithm (for the specified range) is also simple: 

\medskip
\noindent{\bf Algorithm~A1.} 
\begin{itemize} \itemsep 1pt
\item[] {\sc Step 1:} Choose an arbitrary subset of $2i+j+1$ elements from the given $n$
and group them into $m= i + \lfloor \frac{j+1}{2} \rfloor $ pairs and at most one leftover element
(if $j$ is even).

\item[] {\sc Step 2:} Perform the $m$ comparisons and include the largest from each pair in a pool
of $m$ elements. Add the leftover element, if any, to the pool.

\item[] {\sc Step 3:} Select and return the $(i+1)$-th largest from this pool.
  
\end{itemize}

\begin{figure}[htbp]
\centering
\includegraphics[scale=0.85]{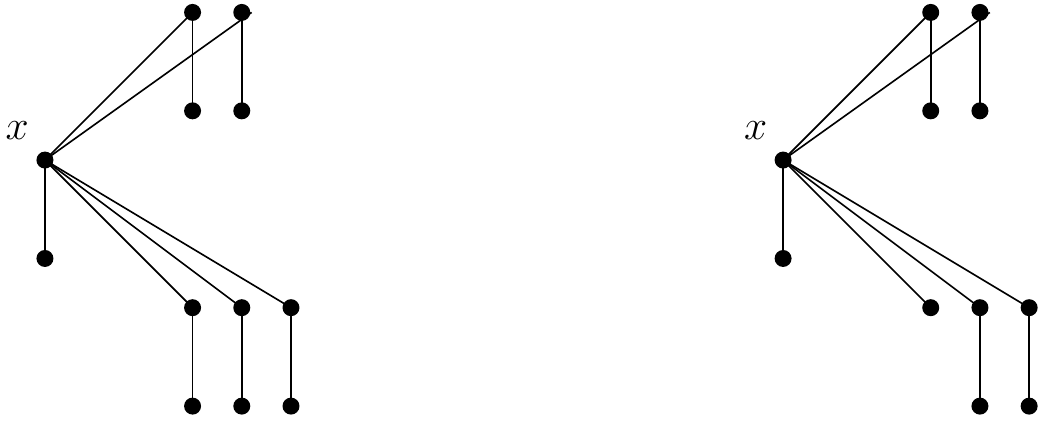}
\caption{Left: Illustration of Algorithm~A1 for selecting a $(2,7)$-mediocre element
  out of $n=12$ elements (left) and a $(2,6)$-mediocre element out of $n=11$ elements (right). 
  Large elements are at the top of the respective edges.}
\label{fig:f1}
\end{figure}

Let us briefly argue about its correctness and refer to Fig.~\ref{fig:f1}.
Denote the selected element by $x$. Assume first that $n$ is even.
On one hand, $x$ is smaller than $i$ (upper) elements in disjoint pairs;
on the other hand, $x$ is larger than $n-2i-1 \geq j$ (upper and lower) elements in disjoint pairs
by the range assumption. The argument is similar for odd $n$; it turns out that the final place
in the poset diagram where the singleton element ends up is irrelevant.
It follows that the algorithm returns an $(i,j)$-mediocre element, as required.

It should be noted that both algorithms (ours as well as Yao's)
make calls to exact selection, however with different input parameters.
As such, we use the current best deterministic algorithm and corresponding
worst-case bound for (exact) selection available.
In particular, selecting the median can be accomplished with at most $2.95 n$ 
comparisons, by using the algorithm of Dor and Zwick~\cite{DZ99}.

\paragraph{The algorithm of Dor and Zwick.}
Consider the problem of selecting the $\alpha n$-th largest element out of given $n$.
By symmetry one may assume that $0 < \alpha \leq 1/2$. 
If $l \geq 0$ is any fixed integer, this task can be accomplished with at most
$c_{\rm Dor-Zwick} \cdot n + o(n)$ comparisons, where 
\begin{equation} \label{eq:alpha,l}
c_{\rm Dor-Zwick} = c_{\rm Dor-Zwick} (\alpha,l) = 1 +(l+2) \left( \alpha + \frac{1-\alpha}{2^l} \right),
\end{equation}
by using an algorithm due to the same authors~\cite{DZ96}.
By letting
$l =\lfloor \log{\frac1\alpha} + \log{\log{\frac1\alpha}} \rfloor$ in
Equation~\eqref{eq:alpha,l}, the authors obtain the following upper bound:
\begin{align} 
c_{\rm Dor-Zwick} (\alpha)  &\leq 1 + \left(\log{\frac1\alpha} + \log{\log{\frac1\alpha}} +2 \right)
\cdot \left( \alpha + \frac{2\alpha(1-\alpha)}{\log{\frac1\alpha}} \right) \label{eq:alpha}\\
&= 1 + \alpha \log{\frac1\alpha} + \alpha \log \log {\frac1\alpha} + O(\alpha). \nonumber
\end{align}
Note that Equations~\eqref{eq:alpha,l} and~\eqref{eq:alpha} only lead to upper bounds
in asymptotic terms.

Next we show that algorithm A1 outperforms Yao's algorithm for finding an
$(\alpha n, \beta n)$-mediocre element
for large $n$ and for a broad range of values of $\alpha$ and suitable $\beta=\beta(\alpha)$,
when using a slightly modified version of the algorithm of Dor and Zwick.
A key difference between our algorithm and Yao's lies in the amount of effort put into
processing the input. Whereas Yao's algorithm chooses an arbitrary subset of elements of
a certain size and ignores the remaining elements, our algorithm looks at more (possibly all) 
input elements and gathers initial information based on grouping the elements into disjoint pairs
and performing the respective comparisons.

\paragraph{Fine-tuning the algorithm of Dor and Zwick.}
For $0<\alpha \leq 1/2$, let $f(\alpha)$ denote the multiplicative constant in
the current best upper bound on the number of comparisons in the algorithm
of Dor and Zwick for selection of the $\alpha n$-th largest
element out of $n$ elements, according to~\eqref{eq:alpha,l}, with one improvement.
Instead of considering only one value for $l$,
namely $l=\lfloor \log{\frac1\alpha} + \log{\log{\frac1\alpha}} \rfloor$,
we also consider the value $l+1$,
and let the algorithm choose the best (\ie, the smallest of the two resulting values
in~\eqref{eq:alpha,l} for the number of comparisons in terms of $\alpha$).
This simple change yields an advantage of Algorithm~A1 over Yao's algorithm for an
extended range of inputs. (Without this change Yao's algorithm wins over A1 for some
small $\alpha$.) 

We first note that the algorithm of Dor and Zwick~\cite{DZ96},
which is a refinement of the algorithm of Sch\"{o}nhage~\etal~\cite{SPP76}, is non-recursive.
As such, the selection target remains the same during its execution, and so choosing the best value
for $l$ can be done at the beginning of the algorithm. (Recall that the seminal algorithm of
Sch\"{o}nhage~\etal~\cite{SPP76} is non-recursive as well.)

To be precise, let
\begin{equation}
  g(\alpha,l) = \left(1 +(l+2) \left( \alpha + \frac{1-\alpha}{2^l} \right) \right). \label{eq:g}
\end{equation}

For a given $0<\alpha \leq 1/2$, let
\begin{align}
  l&=\left\lfloor \log{\frac1\alpha} + \log{\log{\frac1\alpha}} \right\rfloor, \label{eq:l}\\
  f(\alpha) &= \min\left(g(\alpha,l),g(\alpha,l+1)\right). \label{eq:f}
\end{align}

\paragraph{Problem instances and analysis of the number of comparisons.}
Consider the instance $(\alpha n, (1-2\alpha)n -1)$ of the problem of selecting 
a mediocre element, where $\alpha$ is a constant $0<\alpha<1/3$.
The comparison counts for Algorithm~A1 and Algorithm~Yao on this instance are bounded from above
by $c_{\rm A1} \cdot n + o(n)$ and $c_{\rm Yao} \cdot n + o(n)$, respectively, where
%
%
\begin{align}
c_{\rm A1} &= \frac12 \left( 1 + f(2\alpha) \right), \label{eq:alg} \\
c_{\rm Yao} &= (1-\alpha) \cdot f\left(\frac{\alpha}{1-\alpha} \right). \label{eq:yao} 
\end{align}

Indeed, Algorithm~A1 performs $\alpha n + (n/2 -\alpha n)= n/2$ initial comparisons
followed by a selection problem with a fraction $\alpha'=2 \alpha$ from the $n/2$ available.
The element returned by Yao's algorithm corresponds to a selection problem
with a fraction $\alpha'=\frac{\alpha}{1-\alpha}$ from the $(1-\alpha)n$ available.

Since Dor and Zwick~\cite{DZ99} managed to reduce the $3n +o(n)$ comparison bound
to about $2.95 n$, the expression of $f(\alpha)$ in~\eqref{eq:f} can be replaced by
\begin{equation} \label{eq:f1}
  f(\alpha) = \min\left(g(\alpha,l),g(\alpha,l+1),3\right), 
 \end{equation}
or even by
\begin{equation} \label{eq:f2}
  f(\alpha) = \min\left(g(\alpha,l),g(\alpha,l+1),2.95\right). 
 \end{equation}

We next show that Algorithm~A1 outperforms Algorithm~Yao with respect to the
(worst-case) number of comparisons in selecting a mediocre element for $n$ large enough and
for all  instances $(\alpha_s n, (1-2\alpha_s)n -1)$,
where $\alpha_s =s/100$, and $s=1,\ldots,33$; that is, for all percentiles  $s=1,\ldots,33$
and suitable values of the second parameter. This is proven by the data in the two tables
in Fig.~\ref{fig:tables}, where the entries are computed using Equations~\eqref{eq:alg}
and~\eqref{eq:yao}, respectively. Moreover, the results remain the same, regardless
of whether one uses the expression of $f(\alpha)$ in~\eqref{eq:f1} or~\eqref{eq:f2};
to avoid the clutter, we only included the results obtained by using  the expression
of $f(\alpha)$ in~\eqref{eq:f1}.

Note that the computation in~\eqref{eq:alg} may need the value of $f$ for arguments $x>1/2$;
and such values are computed from $f(x)=f(1-x)$ by the symmetry assumption
(\eg, $f(0.6)=f(0.4)=3$, $f(0.8)=f(0.2)=2.5$). 
The functions $c_{\rm A1}(\alpha)$ and $c_{\rm Yao}(\alpha)$  
for $\alpha \in (0,1/3)$, as given by~\eqref{eq:alg} and~\eqref{eq:yao}, are plotted in
Fig~\ref{fig:f2}; however, a proof that  $c_{\rm A1}(\alpha) < c_{\rm Yao}(\alpha)$ on this interval
is missing. 
\smallskip

\begin{figure}[htbp]
\begin{minipage}[b]{0.61\linewidth}
\centering

{\small
  \begin{tabular}{ccccc} \toprule
$\alpha_s$ & $l$ & $g(\alpha_s,l)$ & $g(\alpha_s,l+1)$ & $f(\alpha_s)$ \\ \midrule
$0.01$ & $9$ & $1.1312$ & $1.1316$ & $1.1312$ \\ 
$0.02$ & $8$ & $1.2382$ & $1.2410$ & $1.2382$ \\ 
$0.03$ & $7$ & $1.3382$ & $1.3378$ & $1.3378$ \\ 
$0.04$ & $6$ & $1.4400$ & $1.4275$ & $1.4275$ \\ 
$0.05$ & $6$ & $1.5187$ & $1.5168$ & $1.5168$ \\ 
$0.06$ & $6$ & $1.5975$ & $1.6060$ & $1.5975$ \\ 
$0.07$ & $5$ & $1.6934$ & $1.6762$ & $1.6762$ \\ 
$0.08$ & $5$ & $1.7612$ & $1.7550$ & $1.7550$ \\ 
$0.09$ & $5$ & $1.8290$ & $1.8337$ & $1.8290$ \\ 
$0.10$ & $5$ & $1.8968$ & $1.9125$ & $1.8968$ \\ 
$0.11$ & $4$ & $1.9937$ & $1.9646$ & $1.9646$ \\ 
$0.12$ & $4$ & $2.0500$ & $2.0325$ & $2.0320$ \\ 
$0.13$ & $4$ & $2.1062$ & $2.1003$ & $2.1003$ \\ 
$0.14$ & $4$ & $2.1625$ & $2.1681$ & $2.1625$ \\ 
$0.15$ & $4$ & $2.2187$ & $2.2359$ & $2.2187$ \\ 
$0.16$ & $4$ & $2.2750$ & $2.3037$ & $2.2750$ \\ 
$0.17$ & $3$ & $2.3687$ & $2.3312$ & $2.3312$ \\ 
$0.18$ & $3$ & $2.4125$ & $2.3875$ & $2.3875$ \\ 
$0.19$ & $3$ & $2.4562$ & $2.4437$ & $2.4437$ \\ 
$0.20$ & $3$ & $2.5000$ & $2.5000$ & $2.5000$ \\ 
$0.21$ & $3$ & $2.5437$ & $2.5562$ & $2.5437$ \\ 
$0.22$ & $3$ & $2.5875$ & $2.6125$ & $2.5875$ \\ 
$0.23$ & $3$ & $2.6312$ & $2.6687$ & $2.6312$ \\ 
$0.24$ & $3$ & $2.6750$ & $2.7250$ & $2.6750$ \\ 
$0.25$ & $2$ & $2.7500$ & $2.7187$ & $2.7187$ \\ 
$0.26$ & $2$ & $2.7800$ & $2.7625$ & $2.7625$ \\ 
$0.27$ & $2$ & $2.8100$ & $2.8062$ & $2.8062$ \\ 
$0.28$ & $2$ & $2.8400$ & $2.8500$ & $2.8400$ \\ 
$0.29$ & $2$ & $2.8700$ & $2.8937$ & $2.8700$ \\ 
$0.30$ & $2$ & $2.9000$ & $2.9375$ & $2.9000$ \\ 
$0.31$ & $2$ & $2.9300$ & $2.9812$ & $2.9300$ \\ 
$0.32$ & $2$ & $2.9600$ & $3.0250$ & $2.9600$ \\ 
$0.33$ & $2$ & $2.9900$ & $3.0687$ & $2.9900$ \\  \bottomrule
\end{tabular}
}

 \end{minipage}
\hspace{0.04\linewidth}
\begin{minipage}[b]{0.3\linewidth}
\centering

{\small
  \begin{tabular}{ccc} \toprule
$\alpha_s$ & $c_{\rm A1}(\alpha_s)$  & $c_{\rm Yao}(\alpha_s)$ \\ \midrule
$0.01$ & $1.1191$ & $1.1210$ \\ 
$0.02$ & $1.2137$ & $1.2175$ \\ 
$0.03$ & $1.2987$ & $1.3069$ \\ 
$0.04$ & $1.3775$ & $1.3846$ \\ 
$0.05$ & $1.4484$ & $1.4625$ \\ 
$0.06$ & $1.5162$ & $1.5300$ \\ 
$0.07$ & $1.5812$ & $1.5975$ \\ 
$0.08$ & $1.6375$ & $1.6637$ \\ 
$0.09$ & $1.6937$ & $1.7193$ \\ 
$0.10$ & $1.7500$ & $1.7750$ \\ 
$0.11$ & $1.7937$ & $1.8306$ \\ 
$0.12$ & $1.8375$ & $1.8850$ \\ 
$0.13$ & $1.8812$ & $1.9275$ \\ 
$0.14$ & $1.9200$ & $1.9700$ \\ 
$0.15$ & $1.9500$ & $2.0125$ \\ 
$0.16$ & $1.9800$ & $2.0550$ \\ 
$0.17$ & $2.0000$ & $2.0925$ \\ 
$0.18$ & $2.0000$ & $2.1200$ \\ 
$0.19$ & $2.0000$ & $2.1475$ \\ 
$0.20$ & $2.0000$ & $2.1750$ \\ 
$0.21$ & $2.0000$ & $2.2025$ \\ 
$0.22$ & $2.0000$ & $2.2200$ \\ 
$0.23$ & $2.0000$ & $2.2300$ \\ 
$0.24$ & $2.0000$ & $2.2400$ \\ 
$0.25$ & $2.0000$ & $2.2500$ \\ 
$0.26$ & $2.0000$ & $2.2200$ \\ 
$0.27$ & $2.0000$ & $2.1900$ \\ 
$0.28$ & $2.0000$ & $2.1600$ \\ 
$0.29$ & $2.0000$ & $2.1300$ \\ 
$0.30$ & $2.0000$ & $2.1000$ \\ 
$0.31$ & $2.0000$ & $2.0700$ \\ 
$0.32$ & $2.0000$ & $2.0400$ \\ 
$0.33$ & $2.0000$ & $2.0100$ \\ \bottomrule
\end{tabular}
}

\end{minipage}
\caption{Left: the values of $f(\alpha_s)$, $\alpha_s =s/100$, $s=1,\ldots,33$, for the algorithm
  of Dor and Zwick. Note that $f(\alpha)=3$ for $1/3 \leq \alpha \leq 1/2$. 
  Right: the comparison counts per element of A1 versus Yao on instances
  $(\alpha_s n, (1-2\alpha_s)n -1)$, where $\alpha_s =s/100$ and $s=1,\ldots,33$
  (rounded to four decimals).}
\label{fig:tables}
\end{figure}
%

%
\begin{figure}[htbp]
\centering
\includegraphics[scale=0.99]{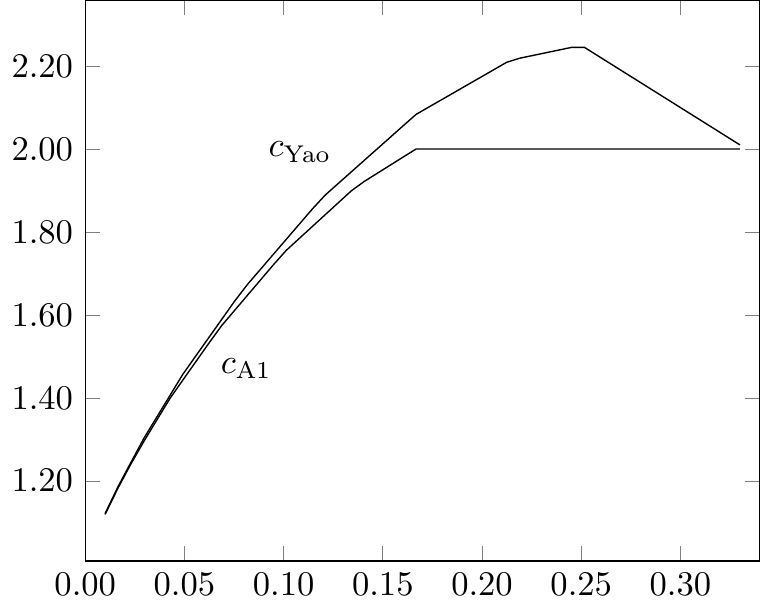}
\caption{$c_{\rm A1}$ vs. $c_{\rm Yao}$.}
\label{fig:f2}
\end{figure}

\paragraph{Extension to hyperpairs.}
It turns out that Algorithm A1 is not a singular case but a member of larger family.
Instead of working with pairs,  a more general scheme  that works with hyperpairs
(algorithm A below) can be designed. (See~\cite{CD20,DZ96,SPP76} for other uses
of hyperpairs in selection.)
The algorithm picks a group size that is a power of $2$
and divides the elements into groups of that size. It then computes the largest in each group
by a tournament method and then calls a suitable selection procedure on these set of largest
group elements. We omit a formal algorithm description but provide a figure instead
(Fig.~\ref{fig:f3}). 
\begin{figure}[htbp]
\centering
\includegraphics[scale=0.85]{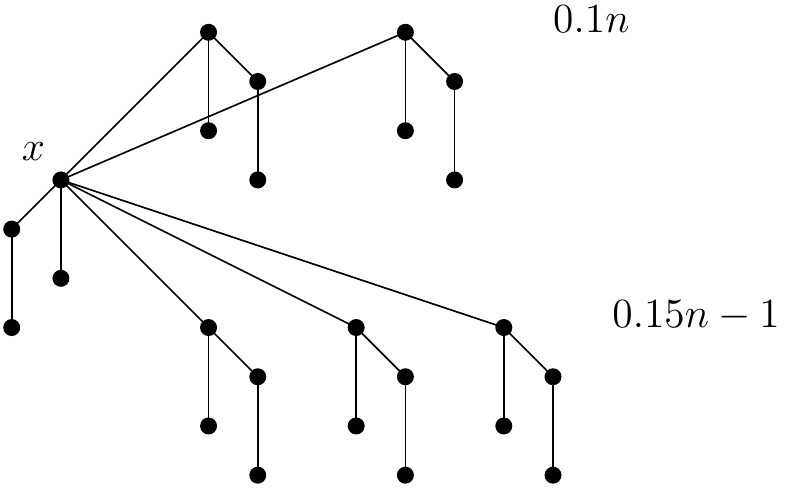}
\caption{Left: Illustration of Algorithm~A working with groups of $4$
  for selecting a $(2,15)$-mediocre element out of $n=24$ elements.
  For large $n$ this scheme finds a $(0.1 n, 0.6 n -1)$-mediocre element out of $n$.}
\label{fig:f3}
\end{figure}

Consider the instance $(\alpha n, (1-4\alpha)n -1)$ of the problem of selecting 
a mediocre element, where $\alpha$ is a constant $0<\alpha \leq 1/5$.
Suppose that Algorithm~A works with groups of $4$ elements.
Analogous to~\eqref{eq:alg} and~\eqref{eq:yao}, 
the comparison counts for Algorithm~A and Algorithm~Yao on this instance are bounded from above
by $c_{\rm A} \cdot n + o(n)$ and $c_{\rm Yao} \cdot n + o(n)$, respectively, where
\begin{align}
c_{\rm A} &= \frac14 \left( 3 + f(4\alpha) \right), \label{eq:h-alg} \\
c_{\rm Yao} &= (1-3 \alpha) \cdot f\left(\frac{\alpha}{1- 3 \alpha} \right). \label{eq:h-yao} 
\end{align}

We have $c_{\rm A}< c_{\rm Yao}$ for each of the percentiles $9$ through $16$
  (\ie, $\alpha_s =s/100$, $s=9,\ldots,16$). 
For example: if $\alpha=0.1$, a $(0.1n , 0.6n -1)$-mediocre element is desired;
then $c_{\rm A}=1.5$ and $ c_{\rm Yao}=1.525$.
 If $\alpha=0.13$, a $(0.13n , 0.48n -1)$-mediocre element is desired;
then $c_{\rm A}=1.5$ and $ c_{\rm Yao}=1.56$.
These results are obtained by using  the expression of $f(\alpha)$ in~\eqref{eq:f1}.

\section{Randomized approximate selection} \label{sec:randomized}

Consider the problem of selecting an $(i,j)$-mediocre element where $i+j + 2(i+j)^{3/4} \leq n$
and $i+j=\omega(1)$. In particular, assume that $i+j \geq 16$. 

\paragraph{Algorithms.}
We next specify our algorithm and then compare it with Yao's algorithm running on the same input. 
Our algorithm relies on random sampling and is similar to the Floyd-Rivest randomized algorithm
for selection~\cite{FR75}. For example, if $i=j=n/2 - n^{3/4}$,  an element in the near vicinity
of the median is sought. In this case, the algorithm also resembles the strategies used by
quicksort algorithms that attempt to optimize the sample size for obtaining balanced
partitions~\cite{MR01,MT95}. 

\bigskip
\noindent{\bf Algorithm~A2.} \\
\textbf{Input:} A set $S$ of $n$ elements over a totally ordered universe and a pair $i,j$
where $i+j +2(i+j)^{3/4} \leq n$.\\
\textbf{Output:} An $(i,j)$-mediocre element.

\begin{itemize} \itemsep 1pt
\item[] {\sc Step 0:} Choose an arbitrary subset $S'$ of $m:=i+j +2(i+j)^{3/4}$ elements from
  the given $n$. (Note that $i+j \geq m/2$ by the assumption.)
  
\item[] {\sc Step 1:} Pick a (multi)-set $R$ of $m^{3/4}$ elements in $S'$, chosen
  uniformly and independently at random with replacement. 

\item[] {\sc Step 2:} Let $k =j m^{-1/4} + m^{1/2}/2$.
  Let $x$ be the $k$-th smallest element of $R$,
 computed by a linear-time deterministic selection algorithm.

\item[] {\sc Step 3:} Compare each of the remaining elements of $S' \setminus R$ to $x$.

\item[] {\sc Step 4:} If there are at least $i$ elements of $S'$ larger than $x$
  and at least $j$ elements of $S'$ smaller than $x$ return $x$, otherwise FAIL. 
  
\end{itemize}

\smallskip
\noindent Since $i+j \geq 16$ we have $2 (i+j)^{3/4} \leq i+j$ or $i+j \geq m/2$ and thus
\[ 2 (i+j)^{3/4} \geq 2 \left( \frac{m}{2} \right)^{3/4} \geq m^{3/4}. \]
In particular $j$ is bounded from above as follows
\[ j= m- i - 2 (i+j)^{3/4} \leq m - i - m^{3/4}. \]
Note that $k$ in {\sc Step 2} satisfies
\[ m^{1/2}/2 \leq k \leq (m-2(i+j)^{3/4})m^{-1/4} + m^{1/2}/2 \leq m^{3/4} - m^{1/2}/2. \]
Observe that (i)~Algorithm A2 performs at most $m + O(m^{3/4})$ comparisons;
and (ii)~it either correctly outputs an $(i,j)$-mediocre element or FAIL. 

\paragraph{Analysis of the number of comparisons.}
Our analysis is an adaptation of that of  the classic randomized algorithm for selection;
see~\cite{FR75}, but also~\cite[Sec.~3.3]{MR95} and~\cite[Sec.~3.4]{MU05}. In particular,
the randomized selection algorithm and Algorithm~A2 both fail for similar reasons.

For $r=1,\ldots,m^{3/4}$, define random variables $X_r$ and $Y_r$ by
\begin{align*}
X_r &=
\left\{ \begin{array}{ll}
1 & \text{if the rank in } S' \text{ of the } r\text{th sample is at most } j, \\
0 & \text{else}. \end{array} \right.\\
Y_r &=
\left\{ \begin{array}{ll}
1 & \text{if the rank in } S' \text{ of the } r\text{th sample is at least } m-i+1, \\
0 & \text{else}. \end{array} \right.
\end{align*}

The variables $X_r$ and $Y_r$ are independent, since the sampling is done with replacement.
It is easily seen that
\[ p:=\Prob(X_r=1) =\frac{j}{m} \text{ and } q:=\Prob(Y_r=1) =\frac{i}{m}. \]

Let $X =\sum_{r=1}^{m^{3/4}} X_r$ and $Y =\sum_{r=1}^{m^{3/4}} Y_r$ 
be the random variables counting the number of samples in $R$ of rank at most $j$
and at least $m-i+1$, respectively.
By the linearity of expectation, we have
\begin{align*}
\E[X] &=\sum_{r=1}^{m^{3/4}} \E[X_r] = m^{3/4} p = j m^{-1/4}, \\
\E[Y] &=\sum_{r=1}^{m^{3/4}} \E[Y_r] = m^{3/4} q = i m^{-1/4}.
\end{align*}

Observe that the randomized algorithm A2 fails if and only if the rank of $x$ in $S'$
is outside the interval $[j+1,m-i]$, \ie,
the rank of $x$ is at most $j$ or at least $m-i+1$.
Note that if algorithm A2 fails then at least
$ j m^{-1/4} + m^{1/2}/2$ elements of $R$ have rank at most $j$ or
at least
\begin{align*}
|R| - k &= m^{3/4} - (j m^{-1/4} + m^{1/2}/2) \geq m^{3/4} - (m- m^{3/4} -i)m^{-1/4} - m^{1/2}/2 \\
&= i m^{-1/4} + m^{1/2}/2
\end{align*}
elements of $R$ have rank at least $m-i+1$.
Denote these two bad events by $E_1$ and $E_2$, respectively.
We next bound from above their probability.
(Sharper bounds on the failure probability can be obtained by using Chernoff
bounds~\cite[Ch.~4]{MU05}; however, they do not affect the asymptotics of our algorithm.)

\begin{lemma} \label{lem:FAIL}
\[ \Prob(E_1), \Prob(E_2) \leq m^{-1/4}. \]
\end{lemma}
\begin{proof}
  Since $X_r$ is a Bernoulli trial, $X$ is a binomial random variable with
  parameters $m^{3/4}$ and $p$. Similarly $Y$ is a binomial random variable
  with parameters $m^{3/4}$ and $q$. 

  Observing that $x(1-x) \leq 1/4$ for every $x \in [0,1]$, it follows
  (see for instance~\cite[Sec.~3.2.1]{MU05}) that
\begin{align*}
\Var(X) &= m^{3/4} p(1-p) \leq m^{3/4}/4, \\
\Var(Y) &=m^{3/4} q(1-q) \leq m^{3/4}/4.
\end{align*}

Applying Chebyshev's inequality \eqref{eq:C} to $X$ yields
\begin{align*}
  \Prob(E_1) &\leq \Prob \left(X \geq j m^{-1/4} + m^{1/2}/2 \right) 
  \leq \Prob\left(|X -\E[X]| \geq m^{1/2}/2 \right) \\
  &\leq \frac{\Var(X)}{m/4} \leq \frac{m^{3/4}/4}{m/4} = m^{-1/4}.
\end{align*}

Similarly, applying Chebyshev's inequality \eqref{eq:C}  to $Y$ yields
\begin{align*}
  \Prob(E_2) &\leq \Prob\left(Y \geq i m^{-1/4} + m^{1/2}/2 \right) 
  \leq \Prob\left(|Y -\E[Y]| \geq m^{1/2}/2 \right) \\
  &\leq \frac{\Var(Y)}{m/4} \leq \frac{m^{3/4}/4}{m/4} = m^{-1/4}.
\end{align*}
The two inequalities have been proved.
\end{proof}

By the union bound, the probability that one execution of Algorithm~A2 fails
is bounded from above by
\[ \Prob(E_1 \cup E_2) \leq \Prob(E_1) + \Prob(E_2) \leq 2 m^{-1/4}. \]

As in~\cite[Sec~3.4]{MU05}, Algorithm~A2 can be converted (from a Monte Carlo algorithm)
to a Las Vegas algorithm by running it repeatedly until it succeeds. 
By Lemma~\ref{lem:FAIL}, the FAIL probability is significantly small,
and so the expected number of comparisons of the resulting algorithm is still $m + o(m)$.
Indeed, the expected number of repetitions until the algorithm succeeds is
at most
\[ \frac{1}{1-2 m^{-1/4}} \leq 1 + O\left(m^{-1/4}\right). \]
Since the number of comparisons in each execution of the algorithm is
$m + O\left(m^{3/4}\right)$, the expected number of comparisons  until success
is at most
\[ \left( 1 + O\left(m^{-1/4}\right) \right) \left( m + O\left(m^{3/4} \right) \right) =
m + O(m^{3/4}). \]

\medskip
We now analyze the average number of comparisons done by Yao's algorithm.
On one hand, the $k$-th largest element out of $n$ given can be found using at most
$n + \min(k,n-k) + o(n)$ comparisons on average~\cite{FR75}.
On the other hand, this task requires $n + \min(k,n-k) + o(n)$ comparisons
on average~\cite{CM89}. Consequently, Yao's algorithm performs $i+j + \min(i,j) + o(i+j)$
comparisons on average.

\paragraph{Comparison.}
Consider the problem of selecting an $(i,j)$-mediocre element where $i+j +$ \linebreak
$2 (i+j)^{3/4} \leq n$
and $i+j \geq \left( \frac23 + \eps \right) n$ for some constant $\eps>0$.  
Algorithm~A2 performs $n +o(n)$ comparisons on average. If $i \approx j$, Yao's algorithm
performs $\frac32 \left( \frac23 + \eps \right) n + o(n) = n + \frac32 \eps n + o(n)$ 
comparisons on average, strictly more than Algorithm~A2 for large $n$. 

For example, let $i=j=n/2 - n^{3/4}$. Whereas  Algorithm~A2 performs $n +o(n)$ comparisons
on average, Yao's algorithm performs $3n/2 + o(n)$ comparisons on average. Indeed, 
the median of $i+j+1=n - 2n^{3/4} +1$ elements can be found in at most $3n/2 + o(n)$
comparisons on average; and the main term in this expression cannot be improved. 

Let $\alpha,\beta>0$ be two constants, where $\alpha+\beta<1$. Algorithm A2 can find
an $(\alpha n, \beta n)$-mediocre element out of $n$, for large $n$,
in $n +o(n)$ comparisons on average, whereas Yao's algorithm performs
$(\alpha + \beta)n + \min(\alpha, \beta) n + o(n)$ comparisons on average.
If $\alpha \approx \beta$ and $\alpha + \beta$ is close to $1$ a significant $33\%$
savings results.

\section {Lower bounds} \label{sec:lower-bounds}

We compute lower bounds by leveraging the work of Sch\"onhage on a related problem,
namely partial order production. In the \emph{partial order production} problem,
we are given a poset $P$ partially ordered by $\preceq_1$,
and another set $S$ of $n$ elements with an underlying, unknown,
total order $\preceq_2$; with $|P| \leq |S|$. The goal is to find a monotone injection
from $P$ to $S$ by querying the total order $\preceq_2$ and minimizing the number
of such queries. Alternatively, the partial order production problem can be
(equivalently) formulated with $|P|=|S|$, by padding $P$ with $|S|-|P|$ singleton elements.

This problem was first studied by Sch\"{o}nhage~\cite{Sch76},
who showed by an information-theoretic argument that $C(P) \geq \lceil \log (n!/e(P)) \rceil$, 
where $C(P)$ is the \emph{minimax comparison complexity} of $P$ and 
$e(P)$ is the number of linear extensions (\ie, total orders) of $P$.
Further results on poset production were obtained by Aigner~\cite{Ai81}.
A.~C.-C.~Yao~\cite{Yao89} proved that Sch\"{o}nhage's lower bound can be achieved asymptotically
in the sense that $C(P)= O(\log (n!/e(P)) + n)$, confirming a conjecture of Saks~\cite{Sa85}.

Finding an $(i,j)$-mediocre element amounts to a special case of the partial order production problem,
where $P$ consists of a center element, $i$ elements above it, and $j$ elements below it. 
For applying Sch\"{o}nhage's lower bound we have
\[ e(P) = i! j! (n-i-j-1)! {n \choose i+j+1}. \]
This yields 
\begin{equation} \label{eq:lb-ij}
C(P) \geq \bigg \lceil \log(n!) - \left( \log(n!) + \log{\frac{i!j!}{(i+j+1)!}}\right) \bigg \rceil=
\bigg \lceil \log {\frac{(i+j+1)!}{i!j!}} \bigg \rceil.
\end{equation}

Interestingly enough, the resulting lower bound does not depend on $n$; observe here
the connection with Yao's hypothesis, namely the question on the independence of $S(i,j,n)$ on $n$
mentioned in Section~\ref{sec:intro}. Moreover, since
\[ \frac{(i+j+1)!}{i!j!}= (j+1) {i+j+1 \choose i} =(i+1) {i+j+1 \choose j}, \]
the above lower bound is rather weak, namely $i+j - o(i+j)$.
We remark that this is not unusual for selection problems and note that a lower bound of
$(i+j+1)-1=i+j<n$ for selecting an $(i,j)$-mediocre element is immediate by a connectivity
argument applied to $P$; see also~\cite[Eq.~(1)]{FR75} and~\cite[Lemma 2]{Yao89}.
On the other hand, observe that the coefficients of the linear terms in the upper bounds
in the right table in Fig.~\ref{fig:tables} are all strictly greater than $1$. 

The situation is similar for randomized algorithms but only in part.
Sch\"{o}nhage's lower bound on the minimax comparison complexity of $P$
in the problem of partial order production was extended to \emph{minimean comparison complexity}
by A.~C.-C.~Yao~\cite{Yao89}.
Denoting this complexity by $\overline{C}(P) $, he showed that
$\overline{C}(P)  \geq \lceil \log (n!/e(P)) \rceil$. As such, the lower bound in~\eqref{eq:lb-ij}
holds for randomized algorithms as well. On the other hand, the trivial lower bound $i+j$ mentioned
previously also holds. Consequently, the upper bound in Theorem~\ref{thm:randomized} is optimal
up to lower order terms.

\section{Conclusion} \label{sec:conclusion}

In Sections~\ref{sec:deterministic} and~\ref{sec:randomized}
we presented two alternative algorithms---one deterministic and one randomized---for finding
a mediocre element, \ie, for approximate selection. 

The deterministic algorithm outperforms Yao's algorithm for large $n$ with respect to the worst-case
number of comparisons for about one third of the percentiles (as the first parameter),
and suitable values of the second parameter, using the best known complexity bounds for
exact selection due to Dor and Zwick~\cite{DZ96}.
Moreover, we suspect that this extends to the entire range of
$\alpha \in (0,1/3)$ and suitable $\beta=\beta(\alpha)$ in the problem of selecting an
$(\alpha n, \beta n)$-mediocre element for large $n$. Whether Yao's algorithm can be beaten
by a deterministic algorithm in the symmetric case $i=j$ remains an interesting question. 

The randomized algorithm outperforms Yao's algorithm for large $n$ with respect to the expected 
number of comparisons for the entire range of $\alpha \in (0,1/2)$ in the problem
of finding an $(\alpha n, \alpha n)$-mediocre element for large $n$. As shown in
Section~\ref{sec:randomized}, these ideas can be also used to generate asymmetric
instances (\ie, with $i \neq j$) with a gap.  

\paragraph{Acknowledgments.} The author thanks Jean Cardinal for stimulating discussions on the
topic. In particular, the idea of examining the existent lower bounds for the partial order production
problem is due to him. The author is also grateful to an anonymous reviewer for
constructive comments and exquisite attention to detail. Finally, thanks go to another 
anonymous reviewer for suggesting that our deterministic algorithm is not a singular example;
the extension of the algorithm described at the end of Section~\ref{sec:deterministic} was
inspired by this suggestion.

\end{document}